\documentclass[a4paper,UKenglish,cleveref, autoref, thm-restate, pdfa]{lipics-v2021}
\usepackage{tikz}
\usepackage{float}
\usepackage{xurl}
\usetikzlibrary{arrows.meta,positioning}
\usetikzlibrary{shapes.multipart, calc}
\usetikzlibrary{decorations.markings}
\usetikzlibrary{arrows}
\usepackage[linesnumbered, vlined, ruled]{algorithm2e}
\usepackage{stmaryrd}
\usepackage{adjustbox}
\usepackage{todonotes}
\definecolor{forestgreen}{HTML}{00A64F}
\definecolor{lavender}{HTML}{EC008C}
\definecolor{purple}{rgb}{0.6, 0.2, 0.8}
\usepackage{setspace}
\usepackage{transparent}
\usepackage[inline]{enumitem}

\newtheorem{thm}{Theorem}
\newtheorem{lem}[thm]{Lemma}

\usepackage{mathtools}

\nolinenumbers

\title{A Formal Analysis of Algorithms for Matroids and Greedoids} 

\author{Mohammad Abdulaziz}{King's College London, United Kingdom  \and \url{https://mabdula.github.io/}}{mohammad.abdulaziz@kcl.ac.uk}{https://orcid.org/0000-0002-8244-518X}{}
\author{Thomas Ammer}{King's College London, United Kingdom \and \url{https://toamme.github.io/}}{thomas.ammer@kcl.ac.uk}{https://orcid.org/0009-0001-5301-4620}{}
\author{Shriya Meenakshisundaram}{King's College London, United Kingdom}{shriya.meenakshisundaram@kcl.ac.uk}{https://orcid.org/0009-0009-5816-8264}{}
\author{Adem Rimpapa}{Technische Universität München, Germany}{adem.rimpapa@in.tum.de}{https://orcid.org/0009-0005-4931-7513}{}

\authorrunning{M. Abdulaziz, T. Ammer, S. Meenakshisundaram, A. Rimpapa} 

\Copyright{} 



\category{} 

\relatedversion{} 

\supplementdetails{Isabelle/HOL Formalisation}{https://doi.org/10.5281/zenodo.15758230}

\EventEditors{John Q. Open and Joan R. Access}
\EventNoEds{2}
\EventLongTitle{42nd Conference on Very Important Topics (CVIT 2016)}
\EventShortTitle{CVIT 2016}
\EventAcronym{CVIT}
\EventYear{2016}
\EventDate{December 24--27, 2016}
\EventLocation{Little Whinging, United Kingdom}
\EventLogo{}
\SeriesVolume{42}
\ArticleNo{23}

\lstdefinestyle{inline}{%
    basicstyle=\ttfamily%
}

\lstdefinelanguage{Isabelle}{ morekeywords={for,context,inductive,begin,end,locale,fixes,record,type_synonym,definition,fun,function,primrec,where,lemma,theorem,unfolding,by,
shows,assumes,and,datatype,using,abbreviation,moreover,have,hence,thus,qed,proof,let,ultimately,show,next,in,if, value, thm,corollary}
    , sensitive=true
    , showstringspaces=false
    , framerule=0pt
    , xleftmargin=2em
    , numbers=left
    , numberstyle=\ttfamily\small
    , firstnumber=1
    , stepnumber=2
    , basicstyle=\ttfamily\small
    , backgroundcolor = \color{white}
    , keywordstyle = {\color{blue}}
    , breaklines=true
    , showspaces=false
    , morecomment=[l]{--}
    , morecomment=[s]{(*}{*)}
    , commentstyle=\color{gray}
    , morestring=[b]"
    , literate={↦}{{$\mapsto$}}{1}
               {∧}{{$\wedge$}}{1}
               {×}{{$\times$}}{1}
               {≡}{{$\equiv$}}{1}
               {∀}{{$\forall$}}{1}
               {∃}{{$\exists$}}{1}  {\\<times>}{{$\times$}}{1}
               {\\<and>}{{$\land$}}{1} {\\<inter>}{{$\cap$}}{1}
        {∈}{{$\in$}}{1} {⇒}{{$\Rightarrow$}}{1} {\\<Rightarrow>}{{$\Rightarrow$}}{1} {\\<And>}{{$\bigwedge$}}{1} {\\<forall>}{{$\forall$}}{1}   {\\<in>}{{$\in$}}{1} {\\<exists>}{{$\exists$}}{1}
        {\\<subseteq>}{{$\subseteq$}}{1} {\\<longrightarrow>}{{$\longrightarrow$}}{1} {\\<Longrightarrow>}{{$\Longrightarrow$}}{1}
        {\\<notin>}{{$\not\in$}}{1}
        {λ}{{$\lambda$}}{1} {\\<lambda>}{{$\lambda$}}{1} {::}{{$::$}}{1}
        {\⊆}{{$\subseteq$}}{1} {\\<subset>}{{$\subset$}}{1} {\\<^sub>m}{{$_m$}}{1} {\\<longleftrightarrow>}{{$\longleftrightarrow$}}{3}
        {\\<pi>}{{$\pi$}}{1} {\\<delta>}{{$\delta$}}{1} {⟦}{{$\llbracket$}}{1} {⟧}{{$\rrbracket$}}{1}
        {⟹}{{$\Longrightarrow$}}{3} {\\<not>}{{$\lnot$}}{1} {\\<le>}{{$\le$}}{1} {\\<rightharpoonup>}{{$\rightharpoonup$}}{2}
        {\\<^sub>\\<V>}{{$_{\mathcal V}$}}{1} {\\<lparr>}{{$\llparenthesis$}}{1} {\\<rparr>}{{$\rrparenthesis$}}{1}
        {\\<leftarrow>}{{$\leftarrow$}}{1} {\\<^sub>\\<O>}{{$_{\mathcal O}$}}{1} {\\<^sub>I}{{$_{\texttt{I}}$}}{1}
        {\\<^sub>G}{{$_{\texttt{G}}$}}{2} {\\<phi>}{{$\varphi$}}{1} {\\<Phi>}{{$\Phi$}}{1} {\\<psi>}{{$\psi$}}{1} {\\<Psi>}{{$\Psi$}}{1}
        {\\<^sub>S}{{$_{\texttt S}$}}{1} {\\<inverse>}{{$^{-1}$}}{1} {\\<^sub>O}{{$_{\texttt O}$}}{1} {\\<^bold>\⋀}{{$\bm\bigwedge$}}{1} {⋀}{{$\bigwedge$}}{1}
        {\\<^bold>\\<or>}{{$\bm\lor$}}{1} {\\<^sub>G}{{$_{\texttt G}$}}{1} {\\<Pi>}{{$\Pi$}}{1} {\\<^sub>I}{{$_{\texttt I}$}}{1} {\≠}{{$\neq$}}{1}
        {\\<bottom>}{{$\bot$}}{1} {\\<^sub>+}{{$_\texttt +$}}{1} {\\<^bold>\\<and>}{{$\bm\land$}}{1} {\\<^bold>\\<not>}{{$\bm\lnot$}}{1}
        {\\<^sub>1}{{$_1$}}{1} {\\<subpset>1}{{$\supset$}}{1} {\\<^sub>2}{{$_2$}}{1} {\\<A>}{{$\mathcal A$}}{1} {\\<Turnstile>}{{$\models$}}{2} {\\<^sub>∀}{{$_\forall$}}{1} {\\<noteq>}{{$\neq$}}{1} 
        {\\<^sub>0}{{$_0$}}{1} {\\<tau>}{{$\tau$}}{1}  {\\<^sub>\\<Omega>}{{$_\Omega$}}{1} {\\<^sub>V}{{$_V$}}{1} {\\<^bold>\\<Or>}{{$\bm\bigvee$}}{1}
        {\\<^sub>P}{{$_\texttt P$}}{1} {\\<^sub>X}{{$_\texttt X$}}{1} {⟹}{{$\Longrightarrow$}}{2} {\\<or>}{{$\lor$}}{1} {\\<^sub>\\<pi>}{{$_\pi$}}{1}
        {\\<^sub>s}{{$_s$}}{1} {\\<^sub>t}{{$_t$}}{1} {\\<^sub>a}{{$_a$}}{1} {\\<^sub>r}{{$_r$}}{1} {\\<^sub>t}{{$_t$}}{1} {\\<^sub>e}{{$_e$}}{1} {\\<^sub>n}{{$_n$}}{1} {\\<^sub>d}{{$_d$}}{1} {\\<^sub>i}{{$_i$}}{1} {\\<^sub>v}{{$_v$}}{1} {\\<^sub>j}{{$_j$}}{1} {\\<^sub>b}{{$_b$}}{1} {\∩}{{$\cap$}}{1} {\\<union>}{{$\cup$}}{1} {\\<Union>}{{$\bigcup$}}{1} {\\<^sup>c\\<TTurnstile>\\<^sub>=}{{${}^c\models_=$}}{1}
        {\\<open>}{{<}}{1} {\\<close>}{{>}}{1} {\\<langle>}{{$\langle$}}{1} {\\<rangle>}{{$\rangle$}}{1} {\\<ge>}{{$\ge$}}{1} {\\<^sup>-\\<^sup>1\\<^sub>C}{{$^{\texttt{-1}}_\texttt{C}$}}{2} {\\<^sup>+\\<^sub>C}{{$^\texttt{+}_\texttt{C}$}}{1} {\\<circ>\\<^sub>C}{{$\circ^\texttt C$}}{1} {\\<top>\\<^sub>C}{{$\top_\texttt{C}$}}{2} {\\<bottom>\\<^sub>C}{{$\bot_\texttt{C}$}}{2} {\\<not>\\<^sub>C}{{$\neg_\texttt{C}$}}{2}
        {\\<squnion>\\<^sub>C}{{$\sqcup_\texttt{C}$}}{2} {\\<sqinter>\\<^sub>C}{{$\sqcap_\texttt{C}$}}{2} {\∃\\<^sub>C}{{$\exists_\texttt{C}$}}{2}  {∀\\<^sub>C}{{$\forall_\texttt{C}$}}{2} {=\\<^sub>C}{{$=_\texttt{C}$}}{2} {\\<sigma>}{{$\sigma$}}{2} {\∉}{{$\notin$}}{1} {⊕}{{$\oplus$}}{1} {\\<nexists>}{{$\nexists$}}{1} {\\<setminus>}{{$\setminus$}}{1}
}

\begin{document}

\maketitle

\begin{abstract}
We present a formal analysis, in Isabelle/HOL~\cite{isabelleholref}, of optimisation algorithms for matroids, which are useful generalisations of combinatorial structures that occur in optimisation, and greedoids, which are a generalisation of matroids.
Although some formalisation work has been done earlier on matroids, our work here presents the first formalisation of results on greedoids, and many results we formalise in relation to matroids are also formalised for the first time in this work.
We formalise the analysis of a number of optimisation algorithms for matroids and greedoids.
We also derive from those algorithms executable implementations of Kruskal's algorithm for computing optimal spanning trees, an algorithm for maximum cardinality matching for bi-partite graphs, and Prim's algorithm for computing minimum weight spanning trees.
\end{abstract}

\newcommand{\savespace}[1]{}

\section{Introduction}

Matroids are algebraic structures that generalise the concepts of graphs and matrices, where a matroid is a pair of sets $(E, \mathcal{F})$ satisfying a number of conditions, e.g.\ $E$ is a finite set and $\mathcal{F}\subseteq 2^{E}$.
Although they have a generally rich mathematical structure which inherently justifies studying them, there are two main motivations to studying them in the context of combinatorial optimisation.
First, optimisation problems defined on matroids generalise, albeit not always in a practically useful way, many standard combinatorial optimisation problems: e.g.\ the minimisation problem for matroids over modular objective functions generalises the travelling salesman problem, the shortest-path problem, minimum spanning trees, and Steiner trees, all of which are important standard optimisation problems.
This means that verifying an algorithm for the minimisation problem, for instance, allows for deriving algorithms to solve many other problems (modulo some conditions that we will discuss later), and in many cases those derived algorithms are practical, albeit sometimes asymptotically slower than the fastest possible algorithms.
The second reason to study matroids is that, since they are an abstraction of other concrete structures, like graphs, reasoning about (algorithms that process) them avoids a lot of the combinatorial reasoning and the case analyses encountered when reasoning about (algorithms processing) the more concrete structures.

In this paper we formalise, in Isabelle/HOL~\cite{isabelleholref}, a number of results about matroids, and greedoids, which are a generalisation of matroids.
We formally analyse greedy algorithms for the maximisation problems for both structures and prove guarantees on the optimality of the solutions those greedy algorithms find.
We also formally analyse an algorithm for the matroid intersection problem, where one aims to find the largest member of $\mathcal{F}_1\cap\mathcal{F}_2$, for two matroids $(E, \mathcal{F}_1)$ and $(E, \mathcal{F}_2)$.
We then instantiate the verified algorithms with concrete instances to obtain executable algorithms for standard graph problems: we obtain an $O(n^2\cdot \log^2 n)$ implementation of Kruskal's algorithm from the greedy algorithm for matroids,\footnote{We verify a variant to compute maximum spanning trees.} an $O(n\cdot m\cdot (\log n + \log m))$ algorithm for maximum cardinality bi-partite matching from the matroid intersection algorithm, and an $O(n\cdot m\cdot\log n)$ implementation of Prim's algorithm for minimum weight spanning trees from the greedy algorithm for greedoids, where $n$ and $m$ are the number of vertices and edges of the input graph, respectively.\footnote{We did not verify the running times of these algorithms in Isabelle/HOL.
Comments on running time are an informal analysis only.} 
In addition to verifying these algorithms, we formalise a number of results from the theory of matroids and greedoids, e.g.\ the fact that matroids can be fully characterised by the greedy algorithm, and the relationship between greedoids and antimatroids.
With the exception of deriving Kruskal's algorithm from the greedy algorithm for maximisation over matroids~\cite{Kruskal-AFP} and Edmonds' Intersection Theorem (Theorem~\ref{lemma:ranks})~\cite{lean_matroids}, all the above results were not formalised before in any theorem prover.
In our formalisation, we follow Korte and Vygen's textbook~\cite{KorteVygenOptimisation} and Schrijver's textbook~\cite{schrijverBook}, with the former being our main reference.



\section{Background}

\subparagraph*{Matroids}
A set system $(E, \mathcal{F})$ consists of a \emph{carrier} set $E$ and a collection of \emph{independent sets} $\mathcal{F} \subseteq 2^{E}$.
We only consider finite carriers. We call $(E, \mathcal{F})$ an \emph{independence system} iff (M1) $\emptyset \in \mathcal{F}$ and (M2) $X \subseteq Y$ and $Y \in \mathcal{F}$ implies $X \in \mathcal{F}$. Members of $\mathcal{F}$ are called \emph{independent}.
Sets $F\subseteq E$ with $F \not \in \mathcal{F}$ are called \emph{dependent}. A \emph{matroid} is an independence system satisfying the \emph{augmentation property} (M3): for $X, Y \in \mathcal{F}$ and $|X| > |Y|$, there is $x \in X \setminus Y$ with $Y \cup \lbrace x \rbrace \in \mathcal{F}$.

A \emph{basis} of a set $X \subseteq E$ is an independent subset of $X$ maximal w.r.t.\ set inclusion.
On the contrary, a \emph{circuit} $C \subseteq E$ is a minimal dependent set, e.g.\ there are no circuits if $\mathcal{F} = 2^E$.
For $X \subseteq E$, the \emph{lower rank} $\rho(X)$ and \emph{(upper) rank} $r(X)$ are the cardinalities of the smallest and largest bases of $X$, respectively.
The \emph{rank quotient} $q(E, \mathcal{F})$ is defined as the minimum of $\frac{\rho(X)}{r(X)}$ over all $X \subseteq E$.
For example, if $\mathcal{F}= \lbrace \emptyset, \lbrace 1\rbrace, \lbrace 2\rbrace, \lbrace 3\rbrace, \lbrace 2,3\rbrace \rbrace$ for $E = \lbrace 1,2,3 \rbrace$, $\lbrace 1\rbrace$ is a smallest and $\lbrace 2,3\rbrace$ is a largest basis for $X = E$ and $\frac{\rho(X)}{r(X)} = 1$ for all other $X \subseteq E$ leading to $q(E, \mathcal{F}) = \frac{1}{2}$.
We have $0 < q(E, \mathcal{F}) \leq 1$ in general, and $q(E, \mathcal{F}) = 1$ iff $(E, \mathcal{F})$ is a matroid. Intuitively, the quotient expresses how close an independence system is to being a matroid.

Generalising planar graph duality (see details of instantiation to directed graphs in the next section), we call $(E, \mathcal{F^*})$ the \emph{dual} of an independence system $(E, \mathcal{F})$, where a set $F\in\mathcal{F^*}$ iff there is a basis $B$ of $(E, \mathcal{F})$ s.t.\ $F \cap B = \emptyset$. As a basic property, we have that the dual of the dual $(E, \mathcal{F}^{**})$ is the independence system $(E, \mathcal{F})$ itself. Furthermore, $(E, \mathcal{F}^*)$  is a matroid iff $(E, \mathcal{F})$ is a matroid. For a matroid $(E, \mathcal{F})$ and $F \subseteq E$, $r^*(F) = |F| + r(E \setminus F) - r(E)$, where $r^*$ is the rank of the dual.

\subparagraph*{Formalisation}
Keinholz already formalised independence systems, matroids, bases, circuits, lower ranks and basic properties of these in Isabelle/HOL~\cite{Matroids-AFP} using \emph{locales}.
In Isabelle/HOL, locales~\cite{LocalesBallarin} are named contexts allowing the fixing of constants and the assertion of assumptions on their properties, which will then be available within the locale for definitions and proofs.
Keinholz gives locales \texttt{indep\_system} and \texttt{matroid} shown in Listing~\ref{isabelle:indep_sys_matroid_defs}, each of which fixes a carrier set \texttt{carrier} and an independence predicate \texttt{indep} specifying which subsets of \texttt{carrier} should be independent.
We extend this formalisation with another equivalent definition of the upper rank.
For the rank quotient, we explicitly define the argument of the minimum to be $1$ when $\rho(X) = r(X)$ to treat division by zero and empty carriers.
\begin{figure}[t]
\begin{lstlisting}[
 language=Isabelle,
 caption={Definition of the rank quotient in Isabelle/HOL.},
 label={isabelle:rank_quotient_def},
 captionpos=b,
 numbers=none,
 xleftmargin=0cm,
 columns=flexible,
 linewidth=1.1\linewidth
 ]
definition Frac :: "int \<Rightarrow> int \<Rightarrow> rat" where
  "Frac a b = (if a = b then 1 else Fract a b)"
definition rank_quotient :: "rat" where
  "rank_quotient = Min {Frac (int (lower_rank_of X))  (int (rank X)) | X. X \<subseteq> carrier}"
\end{lstlisting}
\end{figure}

\begin{figure}[t]
\begin{lstlisting}[
 language=Isabelle,
 caption={Two locales with assumptions characterising independent systems and matroids.},
 label={isabelle:indep_sys_matroid_defs},
 captionpos=b,
 numbers=none,
 xleftmargin=0cm,
 columns=flexible
 ]
locale indep_system =
  fixes carrier :: "'a set" fixes indep :: "'a set \<Rightarrow> bool"
  assumes carrier_finite: "finite carrier"
  assumes indep_subset_carrier: "indep X \<Longrightarrow> X \<subseteq> carrier"
  assumes indep_ex: "\<exists>X. indep X"
  assumes indep_subset: "indep X \<Longrightarrow> Y \<subseteq> X \<Longrightarrow> indep Y"

locale matroid = indep_system +
  assumes augment_aux:
    "indep X \<Longrightarrow> indep Y \<Longrightarrow> card X = Suc (card Y) \<Longrightarrow> \<exists> x ∈ X - Y. indep (insert x Y)"
\end{lstlisting}
\end{figure}
\savespace{\vspace{-3ex}}

\subparagraph*{Verification Methodology}
As mentioned earlier, a goal of this work is to formalise some of the theory of matroids and greedoids and then obtain executable algorithms to process them.
We follow Nipkow's~\cite{fdsBook} approach of using locales to model \textit{abstract data types (ADTs)} and to implement data structures such as finite sets.
The functions and assumptions of an ADT can be encapsulated into a locale.
In that approach, one specifies \textit{invariants}, which give the conditions for the data structure to be well-formed, and \textit{abstraction functions}, which convert an instance of the data structure into a value of the abstract mathematical type which it represents (see Listing~\ref{isabelle:custom_set_specs}, for instance).
Implementing an algorithm defined using ADTs can be done by instantiating the ADTs with executable implementations.
The approach thus allows correctness guarantees and a lot of the proofs to be done at an abstract mathematical level, while executable implementations are obtained.

To model iteration, we follow Abdulaziz's~\cite{fdsBook} approach whereby algorithms involving for- or while-loops are specified in Isabelle/HOL as recursive functions manipulating values of algorithm states, which are modelled as records with fields corresponding to the variables which change as the algorithm runs.
The verification of the correctness properties of an algorithm is done using loop invariants, which are shown to hold for the initial state and to be preserved across every distinct execution path of the algorithm.
With a custom-defined induction principle and by providing a termination proof for the function, we can deduce the relevant correctness theorems for the function using the invariant properties.

Our approach has limitations when it comes to the efficiency of executable code, as it cannot generate imperative code, it strikes a good balance of obtaining practical verified implementations for algorithms with deep background mathematical theory.
\savespace{\vspace{-3ex}}

\section{Greedy Algorithm for Matroids}
In order to implement executable algorithms on independence systems and matroids in Isabelle/HOL, we first provide a specification of these structures that allows for executable implementations.
Since the locales \texttt{indep\_system} and \texttt{matroid} work with 'mathematical' sets of the type \texttt{'a set}, we specify as a locale an ADT for sets which can be abstracted to a set of type \texttt{'a set}, which would allow for executability later on after instantiation.
For this, we specify an ADT using a locale extending an existing ADT (\texttt{Set2}) for binary set operations with the function \texttt{cardinality}.
\begin{figure}[t]
\begin{lstlisting}[
 language=Isabelle,
 caption={An ADT for binary set operations: a set implementation is assumed to be of type \texttt{'s}.},
 label={isabelle:custom_set_specs},
 captionpos=b,
 numbers=none,
 xleftmargin=0cm,
 columns=flexible
 ]
locale Card_Set2 = Set2 +
  fixes cardinality :: "'s \<Rightarrow> nat"
  assumes nonempty_repr: "invar X \<Longrightarrow> X \<noteq> empty \<Longrightarrow> set X \<noteq> set empty"
  assumes set_cardinality: "invar X \<Longrightarrow> cardinality X = card (set X)"
\end{lstlisting}
\end{figure}
We then define the locale \texttt{Indep\_System\_Specs}, which includes the specification for \texttt{Card\_Set2}.
Within this locale, an independence system will be represented by a carrier set \texttt{carrier} of type \texttt{'s}, and an independence function \texttt{indep\_fn} of type \texttt{'s $\Rightarrow$ bool}. \texttt{indep\_fn} is an \textit{independence oracle}, that is, it takes a set \texttt{X} as input and returns a Boolean value indicating whether \texttt{X} is independent w.r.t.\ the matroid.

\texttt{Indep\_System\_Specs} defines a predicate \texttt{indep\_system\_axioms} which gives the independence system properties for \texttt{carrier} and \texttt{indep\_fn}, in terms of the well-formedness invariant and the set operations from \texttt{Card\_Set2}.
We also define two abstraction functions \texttt{carrier\_abs} and \texttt{indep\_abs} which convert \texttt{carrier} and \texttt{indep\_fn} to their abstract counterparts of type \texttt{'a set} and \texttt{'a set $\Rightarrow$ bool}, respectively.
We prove that if the invariants hold, then \texttt{indep\_system\_axioms carrier indep\_fn} is equivalent to the abstract independence system properties for the abstractions of \texttt{carrier} and \texttt{indep\_fn}, which allows us to reuse the theory on abstract independence systems for the implementation.
We also assume an invariant that states, among other things, that the independence function must return the same value for any two implementation sets that have the same  abstraction, which is non-vacuous as one abstract set can have multiple different implementations representing it.
This is needed to show that the implementation and abstraction of the independence functions are equivalent.
Finally, we define a locale \texttt{Matroid\_Specs} which has the same specification as \texttt{Indep\_System\_Specs} and contains the additional, analogous predicates and lemmas for matroids.

\begin{figure}[t]
\begin{lstlisting}[
 language=Isabelle,
 caption={A locale specifying an ADT for independence systems.},
 label={isabelle:indep_sys_specs},
 captionpos=b,
 numbers=none,
 xleftmargin=0cm,
 columns=flexible
 ]
locale Indep_System_Specs = set: Card_Set2 + ...
begin
definition indep_system_axioms where
  "indep_system_axioms carrier indep_fn =
    ((\<forall>X. set_inv X \<longrightarrow> indep_fn X \<longrightarrow> subseteq X carrier) \<and>
    (\<exists>X. set_inv X \<and> indep_fn X) \<and>
    (\<forall>X Y. set_inv X \<longrightarrow> set_inv Y \<longrightarrow> indep_fn X \<longrightarrow> subseteq Y X \<longrightarrow> indep_fn Y))"
...
lemma indep_system_abs_equiv:
  "indep_system_axioms carrier indep_fn = 
     indep_system (carrier_abs carrier) (indep_abs indep_fn)"
\end{lstlisting}
\end{figure}

\subsection{Specification of the Best-In-Greedy Algorithm}

\begin{algorithm}[t]
\DontPrintSemicolon 
\SetKwInOut{Input}{Input}
\SetKwInOut{Output}{Output}
Sort $E := \{e_1, \ldots, e_n\}$ such that \(c(e_1) \geq c(e_2) \geq \ldots \geq c(e_n)\);{\color{gray}//sort the carrier set}\;
$F \gets \emptyset$; {\color{gray}//initialise result}\;
\For{$i := 1$ to $n$}{
    \textrm{\color{gray}//check if $F$ stays independent after adding $e_i$, add $e_i$ if it stays independent}\;
    \text{\bf if} {$F \cup \{e_i\} \in \mathcal{F}$} \text{\bf then} {$F \gets F \cup \{e_i\}$;}
}
\Return $F$;\;
\caption{\label{algorithm:best_in_greedy_algo}\label{algorithm:best_in_greedy_algo}\textsc{BestInGreedy}$(E, \mathcal{F}, c)$}
\end{algorithm}

The Best-In-Greedy algorithm is used to solve the maximisation problem on independence systems.
Here, we consider an independence system \((E, \mathcal{F})\) and a nonnegative cost function \(c : E \mapsto \mathbb{R}_{+}\), and want to find a set $X \in \mathcal{F}$ which maximises the total cost $c(X) := \sum_{e \in X} c(e)$.
For the greedy algorithm, we assume the existence of an independence oracle, which given a set $F \subseteq E$, decides whether $F \in \mathcal{F}$ or not. 
The pseudocode of the Best-In-Greedy algorithm is shown in Algorithm~\ref{algorithm:best_in_greedy_algo}.

In order to implement and verify the Best-In-Greedy algorithm, we define a locale \texttt{Best\_In\_Greedy}.
It assumes that its input matroid satisfies \texttt{Matroid\_Specs}.
The elements of the carrier are assumed to be of type \texttt{'a}. Furthermore, the locale fixes a carrier set \texttt{carrier} of type \texttt{'s}, an independence function \texttt{indep\_fn} of type \texttt{'s $\Rightarrow$ bool} and a sorting function \texttt{sort\_desc} of type \texttt{('a $\Rightarrow$ rat) $\Rightarrow$ 'a list $\Rightarrow$ 'a list} which sorts the input list in descending order using the input function as a key.

We define the state for the greedy algorithm as a record type, where \texttt{carrier\_list} is a list of type \texttt{'a list} consisting of the elements of the carrier set (corresponding to $E$) and \texttt{result} is the result set of type \texttt{'s} which is constructed over the course of the algorithm (corresponding to $F$).
With this state type, we can specify the algorithm as a function in Isabelle/HOL as shown in Listing~\ref{isabelle:best_in_greedy_def}.

%
The recursive function \texttt{BestInGreedy} implements the loop in the pseudocode.
It goes recursively through the sorted carrier list, adding elements to the constructed solution as appropriate.
The algorithm requires a cost function \texttt{c} and a list \texttt{order} containing the elements of the carrier set in an initial arbitrary ordering.
These two parameters are not fixed in the locale since we need to quantify over them in the subsequent correctness theorems (both universally and existentially).
We require an explicit initial ordering of the elements for the proof of one of the correctness theorems.
The initial state of the Best-In-Greedy algorithm is defined by setting \texttt{carrier\_list} to \texttt{sort\_desc c order} and \texttt{result} to the empty set.

\subsection{Verification of the Best-In-Greedy Algorithm}

We now describe some of the important aspects of the verification of the Best-In-Greedy algorithm in Isabelle/HOL.
Several of the definitions used in the verification of the Best-In-Greedy algorithm are parametrised by \texttt{c} and \texttt{order}. Whenever these terms appear in a theorem, we assume that \texttt{c} is nonnegative and that the elements of \texttt{order} correspond exactly to those in the carrier set. Additionally, for all the correctness theorems on the greedy algorithm, we assume the predicates \texttt{BestInGreedy\_axioms}, \texttt{sort\_desc\_axioms} and \texttt{matroid.indep\_system\_axioms carrier indep\_fn}. 

\begin{figure}[t]
\begin{lstlisting}[
 language=Isabelle,
 caption={Best-In-Greedy algorithm in Isabelle/HOL and its initial state. Note: for a record \texttt{r}, \texttt{r} $\llparenthesis$ \texttt{x:= v}$\rrparenthesis$ is the same as \texttt{r}, except with the value of \texttt{x} set to \texttt{v}.
},
 label={isabelle:best_in_greedy_def},
 captionpos=b,
 numbers=none,
 xleftmargin=0cm,
 columns=flexible
 ]
function BestInGreedy :: "('a, 's) best_in_greedy_state
  \<Rightarrow> ('a, 's) best_in_greedy_state" where
  "BestInGreedy state = 
    (case (carrier_list state) of [] \<Rightarrow> state
    | (x # xs) \<Rightarrow>  (if indep' (set_insert x (result state)) then
                      let new_result = (set_insert x (result state)) in
                        BestInGreedy (state \<lparr>carrier_list := xs, result := new_result\<rparr>)
                    else BestInGreedy (state \<lparr>carrier_list := xs\<rparr>)))"

definition "initial_state c order =
  \<lparr>carrier_list = (sort_desc c order), result = set_empty\<rparr>"
\end{lstlisting}
\end{figure}

\texttt{BestInGreedy\_axioms} contains the invariants from the locale \texttt{Matroid\_Specs}, whereas \texttt{sort\_desc\_axioms} states that the function \texttt{sort\_desc} sorts the input list in non-increasing order according to the input cost function, and that the sort is stable.

The first important correctness theorem for the Best-In-Greedy algorithm is Lemma~\ref{lemma:best_in_greedy_1}. It states that for any nonnegative cost function $c$ and any $X$ which is a valid solution to the maximisation problem, the cost of the Best-In-Greedy solution is greater than or equal to the rank quotient of the independence system times the cost of solution $X$.

\begin{lem}[Best-In-Greedy cost bound~\cite{Jenkyns1976, Korte1978AnAO}\label{lemma:best_in_greedy_1}]
Let $(E, \mathcal{F})$ be an independence system, with $c : E \to \mathbb{R}_{+}$.
Let $G$ be the output of \textsc{BestInGreedy}.
Then \(c(G) \geq q(E, \mathcal{F}) \cdot c(X)\) for all $X \in \mathcal{F}$.
\begin{proof}
Let $E := \{e_1, e_2, \ldots, e_n\}$, $c : E \to \mathbb{R}_{+}$, and $c(e_1) \geq c(e_2) \geq \ldots \geq c(e_n)$. Let $G_n$ be the final solution found by \textsc{BestInGreedy} assuming that the elements are sorted in the given order, and let $X_n$ be an arbitrary solution. Define $E_j := \{e_1, \ldots, e_j\}$, $G_j := G_n \cap E_j$ and $X_j := X_n \cap E_j$ for $j = 0, \ldots, n$. Set $d_n: = c(e_n)$ and $d_j := c(e_j) - c(e_{j + 1})$ for $j = 1, \ldots, n - 1$.

Since $X_j \in \mathcal{F}$, we have $|X_j| \leq r(E_j)$. Since $G_j$ is a basis of $E_j$, we have $|G_j| \geq \rho(E_j)$. Together with the definition of the rank quotient, we can conclude that
\begin{align*}
c(G_n) &= \sum_{j = 1}^{n} (|G_j| - |G_{j - 1}|) c(e_j)\\
&= \sum_{j = 1}^{n} |G_j| d_j \geq \sum_{j = 1}^{n} \rho(E_j) d_j \geq q(E, \mathcal{F}) \sum_{j = 1}^{n} r(E_j) d_j \geq q(E, \mathcal{F}) \sum_{j = 1}^{n} |X_j| d_j\\
& = q(E, \mathcal{F}) \sum_{j = 1}^{n} (|X_j| - |X_{j - 1}|) c(e_j) = q(E, \mathcal{F}) c(X_n).
\end{align*}
\end{proof}
\end{lem}
Within the appropriate context in the locale \texttt{Best\_In\_Greedy} (which assumes the three axiom predicates described above and that \texttt{c} and \texttt{order} are valid), the lemma can be formulated as in Listing~\ref{isabelle:best_in_greedy_lemma_1} in Isabelle/HOL. Here, \texttt{c\_set c S} denotes the sum of the costs of the elements in a set \texttt{S}, where the cost function is \texttt{c}.
\begin{figure}[t]
\begin{lstlisting}[
 language=Isabelle,
 caption={The first correctness lemma for the Best-In-Greedy algorithm. \texttt{to\_set} is the abstraction function for the set ADT.},
 label={isabelle:best_in_greedy_lemma_1},
 captionpos=b,
 numbers=none,
 xleftmargin=0cm,
 columns=flexible
 ]
lemma BestInGreedy_correct_2:
  "valid_solution X \<Longrightarrow>
    c_set c (to_set (result (BestInGreedy (initial_state c order)))) \<ge> 
    indep_system.rank_quotient (matroid.carrier_abs carrier)
      (matroid.indep_abs indep_fn) * c_set c (to_set X)"
\end{lstlisting}
\end{figure}
In our formalisation, the proof of this lemma followed the informal proof, with the main argument of the proof consisting of the chain of inequalities on sums.
Since we do not explicitly store the current iteration number $j$ in the state, we provide the definition \texttt{num\_iter}, which extracts the iteration number from a given state using the length of the carrier list. The definitions \texttt{carrier\_pref} and \texttt{pref} are used to represent the prefix sets $E_j$, $G_j$ and $X_j$.

A main statement that is often given without further explanation in informal proofs is that for all $j \in \{1, \ldots, n\}$, $G_j$ is a basis of $E_j$. Intuitively, $G_j$ can be seen to be a maximally independent subset of $E_j$ since it is independent by construction, and since no potential candidate elements are skipped during the algorithm. However, in the context of the formalisation, proving this statement requires some more work. In order to be able to use this fact in the proof, we formulate invariant \texttt{invar\_4}, which expresses the desired property.
\begin{figure}[t]
\begin{lstlisting}[
 language=Isabelle,
 caption={The main invariant for the Best-In-Greedy algorithm.},
 label={isabelle:invar_4_def},
 captionpos=b,
 numbers=none,
 xleftmargin=0cm,
 columns=flexible
 ]
definition "invar_4 c order best_in_greedy_state =
  (\<forall>j \<in> {0..(num_iter c order best_in_greedy_state)}.
    (indep_system.basis_in (matroid.indep_abs indep_fn) (carrier_pref c order j)
                            (pref c order (to_set (result best_in_greedy_state)) j)))"
\end{lstlisting}
\end{figure}
For the invariant preservation proofs for \texttt{invar\_4}, defining some more auxiliary invariants and proving they are preserved across the algorithm is necessary. These are not part of the informal proof, since they are either fairly trivial (e.g. showing that \texttt{result} is always a subset of the current carrier prefix) or aspects specific to the formal proof (e.g. showing that \texttt{result} always satisfies the set ADT well-formedness invariant). 
Proving that \texttt{invar\_4} is preserved across the different execution paths of \texttt{BestInGreedy} then boils down to proving that if \texttt{result} is a basis of the current carrier prefix, the result set after one state update will still be a basis of the carrier prefix in the next step. We consider the two recursive cases of \texttt{BestInGreedy}, in which the current element is either added to \texttt{result} if it preserves independence, or left out otherwise.
The two lemmas required for the invariant proofs in these two cases are formulated and proven in the context of the abstract \texttt{indep\_system} locale.
Through the theorems connecting the implementation of independence systems to the abstract theory of independence systems, we are able to use these theorems to complete the invariant preservation proofs of \texttt{invar\_4}.

The second important correctness theorem on the Best-In-Greedy algorithm is Lemma~\ref{lemma:best_in_greedy_2}, which states that there exists a nonnegative cost function $c$ and a valid solution $X$ for which the bound from Lemma~\ref{lemma:best_in_greedy_1} holds.

\begin{lem}[Best-In-Greedy cost bound tightness~\cite{Jenkyns1976, Korte1978AnAO}\label{lemma:best_in_greedy_2}]
Let $(E, \mathcal{F})$ be an independence system.
There exists a cost function $c : E \to \mathbb{R}_+$ and $X\in\mathcal{F}$ s.t.\ for the output $G$ of \textsc{\textsc{BestInGreedy}}, \(c(G) = q(E, \mathcal{F}) \cdot c(X)\).
\begin{proof}
Choose $F \subseteq E$ and bases $B_1, B_2$ of $F$ such that $\frac{|B_1|}{|B_2|} = q(E, \mathcal{F})$. Define
$c(e) := 1$ for $e \in F$, $c(e) := 0$ for $e \in E \setminus F$ and sort $e_1, \ldots, e_n$ such that $c(e_1) \geq c(e_2) \geq \ldots \geq c(e_n)$ and $B_1 = \{e_1, \ldots, e_{|B_1|}\}$. Then $c(G) = |B_1|$ and $c(X) = |B_2|$, which finishes the proof.
\end{proof}
\end{lem}
In our formalisation, this lemma is stated with an additional existential quantifier for the \texttt{order} parameter. The proof in our formalisation proceeds similarly to the informal proof we follow, except that we explicitly construct the initial order list and use the stability of the sorting function to show that the sorting of the elements produces the ordering necessary for the proof.
Additionally, proving that the greedy result is $B_1$ requires defining some new loop invariants (for example that no element of $F - B_1$ is ever in the greedy result) and showing that they are preserved across the algorithm.
\begin{figure}[t]
\begin{lstlisting}[
 language=Isabelle,
 caption={The second correctness lemma for the Best-in-Greedy algorithm.},
 label={isabelle:best_in_greedy_bound_tight},
 captionpos=b,
 numbers=none,
 xleftmargin=0cm,
 columns=flexible
 ]
lemma BestInGreedy_bound_tight:
  "(\<exists>c. nonnegative c \<and> (\<exists>order. valid_order order \<and> (\<exists>X. valid_solution X \<and> c_set c
    (to_set (result (BestInGreedy (initial_state c order)))) = 
    indep_system.rank_quotient (matroid.carrier_abs carrier)
    (matroid.indep_abs indep_fn) * c_set c (to_set X))))"
\end{lstlisting}
\end{figure}

The final correctness theorem on the Best-In-Greedy algorithm concerns the relationship between the performance of the algorithm and matroids:
\begin{thm}[Characterisation of matroids, Edmonds-Rado~\cite{Rado1957, EdmondsMatroids1971}\label{lemma:best_in_greedy_3}]
An independence system $(E, \mathcal{F})$ is a matroid if and only if \textsc{BestInGreedy} finds an optimal solution for the maximisation problem for $(E, \mathcal{F}, c)$ for all cost functions $c : E \to \mathbb{R}_+$.
\end{thm}
\begin{proof}
By Lemma~\ref{lemma:best_in_greedy_1} and Lemma~\ref{lemma:best_in_greedy_2} we have $q(E, \mathcal{F}) < 1$ if and only if there exists a cost function $c : E \to \mathbb{R}_{+}$ for which \textsc{BestInGreedy} does not find an optimum solution. Using the fact that $q(E, \mathcal{F}) < 1$ if and only if $(E, \mathcal{F})$ is not a matroid completes the proof.
\end{proof}

\subparagraph*{Instantiation for Spanning Forests and Oracles} 
For an undirected graph $G$ with a set of edges $E$, a \textit{spanning tree} $T$ is an acyclic subgraph connecting all vertices of $G$ to one another. A \textit{spanning forest} is acyclic and connects all vertices that are connected via $E$.

The set of edges $E$ in undirected graphs forms a matroid under acyclicity, i.e.\ where those sets of edges that form acyclic subgraphs are independent. This is a so-called \textit{graphic matroid} and it is a widely cited example of matroids ~\cite{KorteVygenOptimisation,schrijverBook}. Acyclicity of sets $X \subseteq E$ indeed satisfies the matroid axioms and a circuit would be a cycle in the ordinary sense. 
We formalise this in Isabelle/HOL, including the equivalence of being a basis w.r.t.\ acyclicity and being a spanning forest, making the greedy algorithm suitable to solve the maximum spanning forest problem correctly.
Note that minimum spanning trees can be easily computed using this algorithm by adapting the input.

The instantiation of the Best-In-Greedy algorithm for graph matroids is also known as \textit{Kruskal's Algorithm}~\cite{fc0df122-3305-33a8-bac5-7a6fc3666dfb}. An independence oracle could check for the absence of cycles by a modified Depth-First Search (DFS). Since the independence of the current solution $X$ is maintained as an invariant, it is enough to check that a new element $x$ does not lead to circuits in $X \cup \lbrace x \rbrace$, i.e.\ a new edge $e$ does not add a cycle in the case of the graphical matroid.
We call this a \textit{weak oracle}. It would be simpler to implement and to verify, and it might even allow for running time improvements e.g.\ if available, by using a union-find structure to store connected components of $X$.

We therefore have an extended locale that also specifies the behaviour of a weak oracle, which is that if (a) $X$ is independent and if (b) suitable data structure and auxiliary invariants are satisfied and (c) $x \not \in X$, then the oracle determines whether $X \cup \lbrace x \rbrace$ is independent. 
Subsequently, we verify the greedy algorithm with a weak oracle by proving it equivalent to the first version when called on the empty set.
As a result, correctness can be lifted.
In the end, we obtain an executable algorithm for maximum spanning forests using a simple DFS as weak oracle. A simplified version of the correctness theorem for Kruskal's Algorithm is stated in Listing~\ref{isabelle:kruskal_corr}. {An informal analysis yields that this implementation has a running time of $\mathcal{O}(n^2 \cdot \log^2 n)$. $\mathcal{O}(|E|\cdot \iota)$ is the running time of the abstract algorithm where $\iota$ is the running time of the oracle.}
\begin{figure}
\begin{lstlisting}[
 language=Isabelle,
 caption={Kruskal's algorithm's correctness (simplified).},
 label={isabelle:kruskal_corr},
 captionpos=b,
 numbers=none,
 xleftmargin=0cm,
 columns=flexible
 ]
definition "max_forest X =  (is_spanning_forest(t_set input_G) X \<and> 
    (\<forall> Y. is_spanning_forest (t_set input_G) Y \<longrightarrow> sum c Y \<le> sum c X))"

corollary kruskal_computes_max_spanning_forest:
"max_forest(t_set (result (kruskal input_G (kruskal_init c order))))"
\end{lstlisting}
\end{figure}

\subsection{Greedy Algorithm for Greedoids}

Greedoids are a generalisation of matroids.
Their definition (conditions (M1) and (M3)) is obtained by dropping the condition (M2) from the definition of matroids. 
Similar to matroids, we use a locale to define greedoids, fixing the carrier set, its family of independent subsets, and specifying the axioms of greedoids.
Their intuition is derived from the edge set of undirected trees containing a fixed vertex $r$ of an undirected graph and its total set of edges.
Due to space limits, we briefly survey results on greedoids in our library without expanding on proofs.
Interested readers may refer to the formalisation instead for details. 

We introduce notions of \textit{accessible set systems}, which are set systems in which every independent set contains an element removing which the resulting set continues to be independent, and \textit{antimatroids}, which are accessible set systems that are closed under union.
We formalise two fundamental results~\cite{10.1137/0605024,KorteVygenOptimisation} that are necessary to verify the greedy algorithm  (Theorem~\ref{greedoid:characterisation}), which is the main result in our library on greedoids.
Firstly, given an accessible set system $(E, \, \mathcal{F})$, for $X \, \in \, \mathcal{F}$, $|X| \, = \, k$, there exists an order $x_1, x_2, \, \ldots \, x_k$ for elements of $X$ such that $\forall \, i \, \leq \, k, \, \left\{ x_1, \ldots \, x_i \right\} \, \in \, \mathcal{F}$.
Secondly, a greedoid $(E, \mathcal{F})$ is accessible, as stated in Section 14.1 of \cite{KorteVygenOptimisation}. 
Thirdly, we prove that every antimatroid is a greedoid~\cite{KorteVygenOptimisation}.

We also consider the optimisation of \textit{modular weight functions} $c: 2^E \rightarrow \mathbb{R}$ on greedoids.
Modular weight functions satisfy $c(A \cup B) = c(A) + c(B) - c(A \cap B)$ for all $A, B \subseteq E$.
The algorithm \textsc{GreedoidGreedy} (Listing~\ref{isabelle:greedoid}) keeps track of a current solution $X \in \mathcal{F}$, initially $\emptyset$.
In every iteration, it searches $E$ in the order $e_1,...,e_n$ for a \emph{candidate} $x$ s.t.\ $x \not \in X$ and $X \cup \lbrace x \rbrace \in \mathcal{F}$. It takes the first candidate in $e_1,...,e_n$ such that $c(\lbrace x\rbrace) \geq c(\lbrace y \rbrace)$ for all other candidates $y$. If there is no candidate $x$, the procedure terminates. Otherwise the first best candidate is inserted into $X$, followed by the next iteration. 
Similar to Theorem~\ref{lemma:best_in_greedy_3}, one can characterise certain greedoids with that algorithm.

\begin{thm}
    [Characterisation of Strong-Exchange Greedoids ~\cite{10.1137/0605024}\label{greedoid:characterisation}] 
We say that a greedoid $(E, \mathcal{F})$ has the strong exchange property (SEP) iff for all $A, \, B \, \in \, \mathcal{F}$, $B$ basis w.r.t $\mathcal{F}$, $A \subseteq B$ and $x \, \in E \setminus B$ with $A \cup \lbrace x  \rbrace  \in  \mathcal{F}$, there is $y$ with $A \cup \lbrace y \rbrace \in  \mathcal{F}$ and $(B - \lbrace y \rbrace) \cup \lbrace x \rbrace \in \mathcal{F}$.
 \textsc{GreedoidGreedy} computes a maximum-weight basis in $\mathcal{F}$ for any order of iteration $e_1,...,e_n$ and any modular cost function $c: 2^E \rightarrow \mathbb{R}$ iff $(E, \mathcal{F})$ has the SEP. 
\end{thm}

For the second direction ($\Leftarrow$), the usual proof is by constructing a counterexample to show the contrapositive.
This heavily depends on the order $e_1,...,e_n$ for candidate search, making use of the greedoids-accessibility relationship and order property of accessible set systems.
In Isabelle/HOL (Listing~\ref{isabelle:greedoid}), both costs and the order (formalised as a list) are fixed by a context in which the algorithm is modelled, making $e_1,...,e_n$  an input to the algorithm, as well, similar to the \texttt{order} parameter for \textsc{BestInGreedy}. 

We formalise set systems anew because the existing formalisation of independence systems~\cite{Matroids-AFP} uses a finite carrier set along with the independence predicate and assumes (M2).
Thus, a formalisation of greedoids, which has to drop (M2), cannot be obtained from that. 
Note that, although we prove that a matroid according to Keinholz's formalisation is a greedoid, we formalise a distinct greedy-algorithm for matroids discussed in the last section rather than using \textsc{GreedoidGreedy}.
This is because, {even a matroid with the strong exchange property stated in Theorem~\ref{greedoid:characterisation}}, i.e.\ one for which \textsc{GreedoidGreedy} works, the worst-case running time of \textsc{GreedoidGreedy} is much worse than \textsc{BestInGreedy} -- \textsc{GreedoidGreedy} has a worst-case running time of $\mathcal{O}(|E|^2 \cdot \iota)$, where $\iota$ is the running time of the independence oracle, while \textsc{BestInGreedy} has $\mathcal{O}(|E| \cdot \iota)$.

\subparagraph*{Executability} We give an executable function that is equivalent to the one in Listing~\ref{isabelle:greedoid}. However, candidate search is not performed on all $e \in E$ {since anything in $X$ will never be a next best candidate.} For the current solution $X$, it only checks those $e$ where $e \in E \setminus X$. 

\subparagraph*{Instantiation} An \textit{arborescence} $T$ around a vertex $r$ in a graph with edges $E$ is an acyclic, connected subgraph of $E$ that contains $r$.
If an edge in $E$ is incident on $r$, the set of arborescences around $r$ forms a greedoid. Bases are spanning trees of the connected component of $r$ in $E$. The arborescence greedoid has strong exchange, making the instantiation of \textsc{GreedoidGreedy}, namely, the \textit{Jarnik-Prim Algorithm}~\cite{6773228} optimal.
We obtain a $\mathcal{O}(n \cdot |E|\cdot \log n)$ implementation to compute maximum weight bases where $n$ is the number of vertices in the component around $r$.
{Prim's and Kruskal's algorithm demonstrate that bounds tighter than what can be deduced from the running time of the abstract algorithm, which is $\mathcal{O}(|E|^2 \cdot \iota)$ and $\mathcal{O}(|E| \cdot \iota)$, respectively, is possible for concrete problems.}

\begin{figure}[t!]
\begin{lstlisting}[
 language=Isabelle,
 caption={Formalisation of greedoids, \textsc{GreedoidGreedy}, and the characterisation theorem. },
 label={isabelle:greedoid},
 captionpos=b,
 numbers=none,
 xleftmargin=0cm,
 columns=flexible
 ]

definition "set_system E F = (finite E \<and> (\<forall> X \<in> F. X \<subseteq> E))"

locale greedoid = fixes E :: "'a set" and F :: "'a set set"
  assumes contains_empty_set: "{} \<in> F"
  assumes third_condition: 
    "\<And> X Y. X \<in> F \<Longrightarrow> Y \<in> F \<Longrightarrow> card X > card Y \<Longrightarrow> \<exists>x \<in> X - Y.  Y \<union> {x} \<in> F"
  assumes ss_assum: "set_system E F"

locale greedoid_algorithm = greedoid +  (*Specification of oracles*) begin
context fixes es and c begin
definition "find_best_candidate X= foldr 
  (\<lambda> e acc. if e \<in> X \<or> \<not> orcl e X then acc
              else (case acc of Some d \<Rightarrow> (if c {e} > c {d} then Some e else Some d))
                                | None \<Rightarrow> Some e)     es None"
function greedoid_greedy::"'a list \<Rightarrow> 'a list"  where
  "greedoid_greedy xs = (case  (find_best_candidate (set xs))
      of Some e \<Rightarrow> greedoid_greedy (e#xs) | None \<Rightarrow> xs)"
end

theorem greedoid_characterisation: "strong_exchange_property E F \<longleftrightarrow> 
      (\<forall> c es. valid_modular_weight_func E c \<and>  E = set es \<and> distinct es 
          \<longrightarrow> opt_basis c (set (greedoid_greedy es c Nil)))"
\end{lstlisting}
\end{figure}

\section{Matroid Intersection}
The maximum cardinality matroid intersection problem asks for an $X$ of maximum cardinality in $\mathcal{F}_1 \cap \mathcal{F}_2$ for two matroids $(E, \mathcal{F}_1)$ and $(E, \mathcal{F}_2)$. A formal definition reusing the matroid locale~\cite{Matroids-AFP} is given in Listing~\ref{isabelle:inter_opt_def}. \textsc{BestInGreedy} from the last section adds an element if the solution remains independent. For intersection, adding an element might preserve independence w.r.t.\ $\mathcal{F}_1$ but not $\mathcal{F}_2$. The insight is thus to add an element, then potentially remove an element to maintain independence w.r.t.\ $\mathcal{F}_2$, then insert another element. This three-step process is repeated until an element is added that preserves independence w.r.t.\ both $\mathcal{F}_1$ and $\mathcal{F}_2$. The repetition of this process to improve the solution is called \textit{augmentation}.

The following is an optimality criterion for matroid intersection that involves ranks.
\begin{lem}[Rank Criterion, Edmonds' Intersection Theorem\label{lemma:ranks}~\cite{Edmonds2001SubmodularFM}]
For two matroids $(E, \mathcal{F}_1)$ and $(E, \mathcal{F}_2)$ over the same ground set $E$ with rank functions $r_1$ and $r_2$, respectively, $X \in \mathcal{F}_1 \cap \mathcal{F}_2$, 
and $Q \subseteq E$ it holds that $|X| \leq r_1(Q) + r_2(E \setminus Q)$. Therefore, $|X| \leq r_1(X) + r_2(E \setminus X)$, for any $X\in\mathcal{F}_1 \cap \mathcal{F}_2$.
\end{lem}
\begin{proof}
Both $X \cap Q$ and $X \setminus Q$ are independent in both matroids. This implies $|X \cap Q| \leq r_1(Q)$ and $|X \setminus Q| \leq r_2(E \setminus Q)$. Both inequalities follow from the fact that $|Y| \leq r(Z)$ if $Y \subseteq Z$ and $Y$ independent, which essentially follows from the definition of the rank $r$ in a matroid. Of course, $|X \cap Q| + |X \setminus Q| = |X|$.
\end{proof}

\begin{figure}[t!]
\begin{lstlisting}[
 language=Isabelle,
 caption={Augmentation and Matroid Intersection.},
 label={isabelle:inter_opt_def},
 captionpos=b,
 numbers=none,
 xleftmargin=0cm,
 columns=flexible
 ]
locale double_matroid = matroid1: matroid carrier indep1 
 + matroid2: matroid carrier indep2 
begin
...
definition "is_max X = (indep1 X\<and>indep2 X\<and>\<nexists>Y. indep1 X\<and>indep2 X\<and>card Y > card X)"
definition "A1 = 
{(x, y) | x y. y \<in> carrier - X \<and> x \<in> matroid1.the_circuit (insert y X) - {y}}"
definition "A2 = ..."
definition "S = {y | y. y \<in> carrier - X \<and> indep1 (insert y X)}"
...
context assumes "indep1 X" "indep2 X"
lemma augment_in_both_matroids:
  assumes "\<nexists> q. vwalk_bet (A1 \<union> A2) x q y \<and> length q < length p" "x \<in> S" "y \<in> T" 
   "vwalk_bet (A1 \<union> A2) x p y" "X' = ((X \<union> {p ! i | i. i < length p \<and> even i}) 
              -  {p ! i | i. i < length p \<and> odd i})"
   shows "indep1 X'" and  "indep2 X'" and "card X' = card X + 1" 
...
theorem maximum_characterisation:
  "is_max X  \<longleftrightarrow> \<not> (\<exists> p x y. x \<in> S \<and> y \<in> T \<and> (vwalk_bet (A1 \<union> A2) x p y \<or> x= y))"
\end{lstlisting}
\end{figure}

The rank criterion cannot be exploited immediately for an algorithm since it does not have an immediate computational interpretation. It can be turned into a criterion that is computationally useful using augmentation, which we do in the next section.

\subsection{Augmentation}\label{sec:augment}
Fix a single matroid $(E, \mathcal{F})$.
For $X\in\mathcal{F}$, $x \not \in X$ and $X \cup \lbrace x  \rbrace\in \mathcal{F}$ there is a unique circuit in $X \cup \lbrace x \rbrace$. If there were two of these $C_1, C_2$, there would be a third one $C_3 \subseteq C_1 \cup C_2 \setminus \lbrace x \rbrace \subseteq X$, since both $C_1$ and $C_2$ need to contain $x$. Let $\mathcal{C}(X, x)$ be this unique circuit if $X \cup \lbrace x  \rbrace\in \mathcal{F}$, or $\emptyset$, otherwise. We observe that if $x \in \mathcal{C}(X, y)$ and $X \not \ni y \neq x \in X$, then $x$ can be replaced by $y$, i.e.\ $(X \setminus \lbrace x \rbrace) \cup \lbrace y \rbrace$ is independent. Under certain conditions, this replacement can be repeated as shown by the following lemma:
\begin{lem}[Replacement Lemma~\cite{FRANK1981328}\label{lemma:replacement}]
For a matroid $(E, \mathcal{F})$ and $X \in \mathcal{F}$, assume (1) $x_1, ..., x_s \in X$, (2) $y_1, ..., y_s \in E \setminus X$, (3) $x_k \in \mathcal{C}(X, y_k)$, for all $1\leq k\leq s$, and (4) $x_j \not \in \mathcal{C}(X, y_k)$ for all $1 \leq j < k \leq s$. If (1)-(4) hold, then $(X \setminus \lbrace x_1, ..., x_s\rbrace \cup \lbrace y_1, ..., y_s \rbrace) \in \mathcal{F}$.
\end{lem}
\begin{proof}
We show independence of $X_l := (X \setminus \lbrace x_1, ..., x_l\rbrace) \cup \lbrace y_1, ..., y_l \rbrace$ for $l \leq s$ by induction.
The theorem trivially holds for $l = 0$.
We assume the claim for $X_l$ where $0< l < s$: It might be that $X_l \cup \lbrace y_{l+1}\rbrace$ is independent implying independence of $X_{l+1}$ because of (M2). If $X_l \cup \lbrace y_{l+1}\rbrace$ is dependent, however, it contains a unique circuit $\mathcal{C}(X_l, y_{l+1})$ containing $y_{l+1}$ due to $X_l$'s independence. All deleted $x_1, ..., x_l$ cannot have been part of $\mathcal{C}(X, y_{l+1})$ because of (4) and neither any inserted $y_1, ..., y_{l}$ because of (2), $y_{l+1}$ could, however.  This implies $\mathcal{C}(X, y_{l+1}) \subseteq X_{l+1}$. Therefore $\mathcal{C}(X_l, y_{l+1}) = \mathcal{C}(X, y_{l+1})$. Due to (3), $x_{l+1} \in \mathcal{C}(X, y_{l+1})$, implying independence of $(X_l \setminus \lbrace x_{l+1} \rbrace) \cup \lbrace y_{l+1}\rbrace = X_{l+1}$.
\end{proof}

\noindent The formalisation of the above lemma contains an inductive proof on $s$ within the context of the locale \texttt{matroid}.
The $x_i$s and $y_i$s are formalised as a list of pairs.

From now on, we assume two matroids $(E, \mathcal{F}_1)$ and $(E, \mathcal{F}_2)$ over the same ground set $E$, or formally, we work in the \texttt{double\_matroid} locale (Listing~\ref{isabelle:inter_opt_def}). Following Korte and Vygen~\cite{KorteVygenOptimisation}, for a set $X \in \mathcal{F}_1 \cap \mathcal{F}_2$, we define an auxiliary graph $G_X$ with vertices $E$ and edges $A_{X,1} \cup A_{X,2}$ where $A_{X,1} = \lbrace (x, y) \;.\; y \in E \setminus X \wedge x \in \mathcal{C}_1(X,y) \setminus \lbrace y \rbrace \rbrace$ and $A_{X,2} = \lbrace (y, x) \;.\; y \in E \setminus X \wedge x \in \mathcal{C}_2(X,y) \setminus \lbrace y \rbrace \rbrace$.  $G_X$ is obviously bipartite between $X$ and $E \setminus X$. We define two sets $S_X = \lbrace y.\; y \in E \setminus X \wedge X \cup \lbrace y \rbrace \in \mathcal{F}_1\rbrace$ and
 $T_X = \lbrace y.\; y \in E \setminus X \wedge X \cup \lbrace y \rbrace \in \mathcal{F}_2\rbrace$.
 A path between a vertex from $S_X$ and another from $T_X$ indicates an  alternating sequence of insertion and deletion. Due to bipartiteness, the length will be odd allowing for an augmentation.
We make this precise as follows.

\begin{lem}[Augmentation Lemma\label{lemma:augment}~\cite{KorteVygenOptimisation}]
Let $X \in \mathcal{F}_1 \cap \mathcal{F}_2$ be independent in both matroids and  $p = y_0x_1y_1x_2y_2...x_sy_s$ a shortest path from $y_0\in S_X$ to $y_s\in T_X$.
Then $X' = (X \cup \lbrace y_0, ..., y_s \rbrace) \setminus \lbrace x_1,...,x_s\rbrace$ is independent in both matroids, i.e. $X' \in \mathcal{F}_1\cap \mathcal{F}_2$.
\end{lem}
\begin{proof}
We apply Lemma~\ref{lemma:replacement} to $X\cup \lbrace y_0 \rbrace$, $x_1,..., x_s$ and $y_1,...,y_s$ to show $X' \in \mathcal{F}_1$: For all $0<i \leq s$, $x_{i}$ is the second vertex of the $(2i-1)$th edge of $p$. Because the path alternates between $E\setminus X$ and $X$, the $(2i-1)$th is in $A_{X,2}$ and $x_{i} \in X \cup \lbrace y_0 \rbrace$ for all $x_i$, yielding (1). For all $0< i \leq s$, $y_{i}$ is the second vertex of the $(2i)$th edge in $p$. Because of alternation and pairwise distinctness of the $y$s ($p$ is shortest path), $y_{i} \in E \setminus (X \cup \lbrace y_0 \rbrace)$, implying (2). 
Also, for all $0<i \leq s$, $x_{i}$ is the first vertex of the $(2i)$th edge, which is part of $A_{X,1}$ because of alternation. Therefore, $x_k \in \mathcal{C}_1(X\cup\lbrace y_0\rbrace, y_k)$ for all $1\leq k\leq s$ (3). We assume an $x_{i}$ and $x_j$ where $0 < j <i \leq s$ and $x_j \in \mathcal{C}_1(X\cup\lbrace y_0\rbrace, y_i)$. $x_i$ or $x_j$ is the first vertex of the $2i$th of $2j$th edge of $p$, respectively. $y_j$ and $y_i$ are the second vertices. Both of those edges are part of $A_{X,1}$. As $y_i \in E \setminus X$, and therefore $x_j \neq y_i$, the edge $(x_j, y_i) \in A_{X,1}$. We could then delete $y_j,...,x_{i-1}$ giving a shorter $S_X$-$T_X$-path. Hence (4) is satisfied and Lemma~\ref{lemma:replacement} can be applied.

Analogously, we take $X\cup \lbrace y_s \rbrace$, $x_s,..., x_1$ and $y_{s-1},...,y_0$ to show $X' \in \mathcal{F}_2$.
\end{proof}

\noindent The statement is in Listing~\ref{isabelle:inter_opt_def}. {We show that the edges of paths used by Lemma~\ref{lemma:augment} alternate between $A_1$ and $A_2$ using an existing library on alternating lists for matchings~\cite{blossomIsabelleFull}.
From this we can easily deduce (1)-(3) to apply Lemma~\ref{lemma:replacement}. The re-applicability of matching theory is no coincidence since bipartite matching is an instance of matroid intersection as we show later.}


On top of improvement by augmentation, the absence of an augmenting path characterises maximality of $|X|$ for $X \in \mathcal{F}_1 \cap \mathcal{F}_2$:

\begin{thm}[Optimality Criterion\label{lemma:optcrit}~\cite{KorteVygenOptimisation}]
$X$ is a set of maximum cardinality in $\mathcal{F}_1 \cap \mathcal{F}_2$ iff $G_X$ does not contain a path from some $s\in S_X$ to some $t \in T_X$ (henceforth, an $S_X$-$T_X$ path).
Edmonds' max-min-equality~\cite{Edmonds2001SubmodularFM} is a corollary: $\max \lbrace |X|. \; X \in \mathcal{F}_1 \cap \mathcal{F}_2\rbrace = \min \lbrace r_1(Q) + r_2(E \setminus Q).\; Q \subseteq E\rbrace$.
\end{thm}
\begin{proof}
If there is an $S_X$-$T_X$ path, there is a shortest one as well, which could be used for an augmentation according to Lemma~\ref{lemma:augment}, leading to an increase in the cardinality.

Now assume, that there is no $S_X$-$T_X$ path. We define $R$ as the set of vertices in $G_X$ that are reachable from $S_X$. Obviously $S_X \subseteq R$ and $R \cap T_X = \emptyset$. There are rank functions $r_1$ and $r_2$ w.r.t.\ $\mathcal{F}_1$ and $\mathcal{F}_2$, respectively. 

We prove $r_2(R) = |X \cap R|$ by contradiction: Since $X \cap R$ has to be independent w.r.t.\ the second matroid and therefore $ |X \cap R| = r_2(X \cap R)$, we would have $r_2(X \cap R) < r_2(R)$ since $r_2$ is monotone. Because of the strict inequality, there is $y \in R \setminus X$ where $(X \cap R) \cup \lbrace y \rbrace \in \mathcal{F}_2$. Otherwise, $X \cap R$ would be a basis of $R$ implying equal ranks. As $R \cap T_X = \emptyset$, $X \cup \lbrace y \rbrace \not \in \mathcal{F}_2$. There is $x \in X \setminus R$ with $x \in \mathcal{C}_2(X, y)$ (Otherwise $(X \cap R) \cup \lbrace y \rbrace \not \in \mathcal{F}_2$). By definition, $(y, x) \in A_{X,2}$. That makes $x$ part of $R$, which is a contradiction.

We prove $r_1(E\setminus X) = |X \setminus R|$ by contradiction: Since $X \setminus R$ has to be independent w.r.t.\ the first matroid and therefore $ |X \setminus R| = r_1(X \setminus R)$, we would have $r_1(X \setminus R) < r_1(E \setminus R)$ since $r_1$ is monotone. Because of the strict inequality, there is $y \in (E \setminus R) \setminus X$ where $(X \setminus R) \cup \lbrace y \rbrace \in \mathcal{F}_1$. Otherwise, $X \setminus R$ would be a basis of $E \setminus R$ implying equal ranks. As $S_X \subseteq R$, $X \cup \lbrace y \rbrace \not \in \mathcal{F}_1$. There is $x \in X \cap R$ with $x \in \mathcal{C}_1(X, y)$ (Otherwise $(X \setminus R) \cup \lbrace y \rbrace \not \in \mathcal{F}_1$). By definition, $(x,y) \in A_{X,1}$. That makes $y$ part of $R$, which is a contradiction.

Therefore, $|X| = r_1(E \setminus X) + r_2(X)$. Because of the rank criterion from Lemma~\ref{lemma:ranks}, $X$ satisfies $|X| \leq r_1(Q) + r_2(E \setminus Q)$ with equality. This gives the max-min equality and makes $X$ a set of maximum cardinality that is independent w.r.t.\ both matroids. 
\end{proof}

The formal proof of Theorem~\ref{lemma:optcrit} is split in two directions. The second direction also shows $r_2(R) = |X \cap R|$, $r_1(E\setminus X) = |X \setminus R|$ and $|X| = r_1(E \setminus X) + r_2(X)$ as statements to prove the max-min equality separately. The final formalisation can be seen in Listing~\ref{isabelle:inter_opt_def}.



\subsection{Intersection Algorithm}
We exploit the previous subsection's results for an algorithm~\cite{3c49a919-60de-3f81-bd8d-100cee6b24f2,Lawler1975} that iteratively applies augmenting paths to solve maximum cardinality matroid intersection. 
The main invariant of the algorithm \textsc{MaxMatroidIntersection} is $X \in \mathcal{F}_1 \cap \mathcal{F}_2$ for the current solution $X$ (preservation proof by Lemma~\ref{lemma:augment}). Algorithm~\ref{algo:intersec} shows the pseudocode. In our main reference~\cite{KorteVygenOptimisation}, the circuits $\mathcal{C}_{1/2}(X, y)$, i.e. minimal dependent sets, are computed as explicit sets by iterating over all $x \in X \cup \lbrace y \rbrace$ and taking those where $X \cup \lbrace y \rbrace \setminus \lbrace x \rbrace$ is independent. $G_X$ is then computed by using the definition from above which requires further iterations over the circuits.
Our pseudocode does not need the circuits as an intermediate step and adds an edge $(x,y)$ or $(y,x)$ if $X \cup \lbrace y \rbrace$ is dependent or $X \setminus \lbrace x \rbrace \cup \lbrace y \rbrace$ is independent, respectively.
\begin{algorithm}[t]
\SetAlCapHSkip{0pt}
\SetAlgoHangIndent{0pt}
\SetInd{0.5em}{0.5em}
\SetVlineSkip{0.3mm}
\setlength{\algomargin}{10pt}
\caption{\textsc{MaxMatroidIntersection}($E$, $\mathcal{F}_1$, $\mathcal{F}_2$) \label{algo:intersec}}
Initialise $X \leftarrow \emptyset$;\\
\While{$True$}
{
\textit{compute $G_X$:} Initialise $S_X \leftarrow \emptyset$; $T_X \leftarrow \emptyset$; $A_{X,1} \leftarrow \emptyset$; $A_{X,2} \leftarrow \emptyset$;\\
\For{$y \in E \setminus X$}
{
\textbf{if} $X \cup \lbrace y\rbrace \in \mathcal{F}_1$ \textbf{then} $S_X \leftarrow S_X \cup \lbrace y\rbrace$;\\
\textbf{else}
\textbf{for} $x \in X$ \textbf{do} [ \textbf{if} $X \setminus \lbrace x \rbrace \cup \lbrace y \rbrace \in \mathcal{F}_1$ \textbf{then}
$A_{X,1} \leftarrow A_{X,1} \cup \lbrace (x, y) \rbrace$;]

\textbf{if} $X \cup \lbrace y\rbrace \in \mathcal{F}_2$ \textbf{then} $T_X \leftarrow T_X \cup \lbrace y\rbrace$;\\
\textbf{else}
\textbf{for} $x \in X$ \textbf{do} [ \textbf{if} $X \setminus \lbrace x \rbrace \cup \lbrace y \rbrace \in \mathcal{F}_2$ \textbf{then}
$A_{X,2} \leftarrow A_{X,2} \cup \lbrace (y, x) \rbrace$; ]

}
\uIf{$\exists$ path leading from $S_X$ to $T_X$ via the edges in $A_{X,1} \cup A_{X,2}$}{
find a shortest path $P = y_0x_1y_1...x_sy_s$ leading from $S_X$ to $T_X$;\\
augment along $P$: $X \leftarrow X \cup \lbrace y_0,...,y_s\rbrace \setminus \lbrace x_1,...,x_s\rbrace$;
}
\textbf{else}
\Return $X$ as maximum cardinality set in $\mathcal{F}_1 \cap \mathcal{F}_2$;
}
\end{algorithm}

\begin{figure}[t!]
\begin{lstlisting}[
 language=Isabelle,
 caption={Formalisation of the Intersection Algorithm.},
 label={isabelle:intersec_algo},
 captionpos=b,
 numbers=none,
 xleftmargin=0cm,
 columns=flexible
 ]  
definition "treat1 y X init_map = inner_fold X (\<lambda> x edges. 
   if weak_orcl1 y (set_delete x X)  then add_edge edges x y else edges) init_map"
definition "treat2 y X init_map = ..."
definition "compute_graph X E_without_X =
 outer_fold E_without_X (\<lambda> y (SX, TX, edges). (
  let (SX, TX, edges) = (if weak_orcl1 y X then (insert y SX, TX, edges) 
          else (SX, TX, treat1 y X edges));
      (SX, TX, edges) = (if weak_orcl2 y X then (SX, insert y TX, edges) 
          else (SX, TX, treat2 y X edges))
  in (SX, TX, edges))) (vset_empty, vset_empty, empty)"
  
fun augment where ...

function (domintros) matroid_intersection::"'mset intersec_state\<Rightarrow> 'mset intersec_state"  where
  "matroid_intersection state =
 (let X = sol state; (SX, TX, graph) = compute_graph X (complement X) in
     (case find_path SX TX graph of None \<Rightarrow> state |
      Some p \<Rightarrow> matroid_intersection (state \<lparr> sol := augment X p\<rparr>) ))"
      
definition "initial_state = \<lparr> sol= set_empty \<rparr>"
\end{lstlisting}
\end{figure}
\textsc{MaxMatroidIntersection} is again formalised within a locale that assumes the necessary subprocedures, most notably a function \texttt{find\_path} for path selection. This will take the graph $G_X$ as an adjacency map. We also have two weak independence oracles \texttt{orcl1} and \texttt{orcl2}, one for each matroid. We assume a Haskell-style \texttt{outer\_fold} and \texttt{inner\_fold} to simulate for-loops over sets in the functional language of Isabelle/HOL. We furthermore assume a function to compute set complements as required in Line 4 of Algorithm~\ref{algo:intersec}.

Listing~\ref{isabelle:intersec_algo} shows the formalisation of the for-loops. \texttt{treat1} and \texttt{treat2} simulate Line 6 and Line 8, respectively. The outer for-loop between Lines 4 to 8 is realised by \texttt{compute\_graph}. We also have a simple function \texttt{augment} performing the augmentation, i.e.\ alternating deletion and insertion to a set. The function \texttt{matroid\_intersection} (while loop) calls the graph computation, gives $A_{X,1} \cup A_{X,2}$, $S_X$ and $T_X$ to the path selection function and terminates or calls augmentation according to the result of path selection. \texttt{matroid\_intersection} maps the solution $X$ before entering the loop (intended for Line 1 of Algorithm~\ref{algo:intersec}) to the $X$ after finishing the loop. 
We use a record with a single variable for the current solution.


Again, we use invariants for loop verification. Independence of $X$ is the major invariant and data structure well-formedness invariants are minor ones. As for \textsc{BestInGreedy}, we prove single-step preservation for execution paths which is then lifted to preservation for the whole function by a computational induction. The first lemma in Listing~\ref{isabelle:intersec_invar} says that by executing one step (a) $X$ remains in $\mathcal{F}_1 \cap \mathcal{F}_2$ (by Lemma~\ref{lemma:augment}), (b) $X$ continues to satisfy its data structure invariant, (c) $X \subseteq E$, (d) $X$'s cardinality increases by $1$ (by Lemma~\ref{lemma:augment}) and (e) neither $S_X$ nor $T_X$ are empty. The Listing also contains the lemma certifying optimality if the terminating branch of \texttt{matroid\_intersection} is taken (by Theorem~\ref{lemma:optcrit}). The second lemma together with (a), (b) and (c) can be combined by the induction principle into correctness for general inputs $X$. When (d) and (e) are used, there is termination. Total correctness is proven when \texttt{matroid\_intersection} is called on \texttt{initial\_state}.
\begin{figure}[t!]
\begin{lstlisting}[
 language=Isabelle,
 caption={Invariants, Termination and Total Correctness of the Intersection Algorithm. Note that \texttt{f\_dom input} for a function \texttt{f} means that \texttt{f} terminates on \texttt{input}. },
 label={isabelle:intersec_invar},
 captionpos=b,
 numbers=none,
 xleftmargin=0cm,
  columns=flexible
 ] 
definition "indep_invar state =
   (indep1 (to_set (sol state)) \<and> (indep2 (to_set (sol state))))" 
 
lemma indep_invar_recurse_improvement:
  assumes "matroid_intersection_recurse_cond state" "indep_invar state"
          "set_invar (sol state)" "to_set (sol state) \<subseteq> carrier"
  shows   "indep_invar (matroid_intersection_recurse_upd state)" 
          "set_invar (sol (matroid_intersection_recurse_upd state))"
          "to_set (sol (matroid_intersection_recurse_upd state)) \<subseteq> carrier"
          "card (to_set (sol (matroid_intersection_recurse_upd state)) =
             card (to_set (sol state)) + 1"
          "S (to_set (sol state)) \<noteq> {}"  "T (to_set (sol state)) \<noteq> {}"
          
lemma indep_invar_max_found:
  assumes "matroid_intersection_terminates_cond state" "indep_invar state"...
  shows   "is_max (to_set (sol state))"
  
lemma matroid_intersection_correctness_general:
  assumes "indep_invar state" "set_invar (sol state)" "to_set (sol state) \<subseteq> carrier"
    shows "is_max (to_set (sol (matroid_intersection state)))"
         
lemma matroid_intersection_terminates_general:
  assumes  "indep_invar state" "set_invar (sol  state)" "to_set (sol state) \<subseteq> carrier"
           "m = card carrier - card (to_set (sol state))"
  shows    "matroid_intersection_dom state"
  
lemma matroid_intersection_total_correctness: 
     "is_max (to_set (sol (matroid_intersection initial_state)))" 
     "matroid_intersection_dom initial_state"
\end{lstlisting}
\end{figure}


\subparagraph*{Running Time and Circuit Oracles} The number of iterations to build $G_X$ is $\mathcal{O}(|E|^2)$.
Each matroid comes with an oracle which is assumed to be $\iota_1$ or $\iota_2$, respectively. This results in $\mathcal{O}(|E|^3 \cdot (\iota_1 + \iota_2))$ overall running time. We might have to check $|X|\cdot (|E| - |X|)$ pairs of vertices. If - depending on the problem - the size of the circuits $\sigma_1$ and $\sigma_2$ is assumed to be small, e.g.\ $\mathcal{O}(\log |E|)$ or even $\mathcal{O}(1)$, most of the inner for-loops' iterations were wasted. If the circuits can be computed in time $\kappa_1, \kappa_2 \in o(|E|)$, the time would be the more efficient $\mathcal{O}(|E|^2 \cdot (\iota_1 + \iota_2 + \kappa_1 + \kappa_2 + \sigma_1 + \sigma_2))$.

In Isabelle/HOL, weak circuit oracles are assumed to return data structures storing $\mathcal{C}_{1/2}(X, y) \setminus \lbrace y \rbrace$ provided that $X$ is independent and $X \cup \lbrace y \rbrace$ is dependent. These are added to the locale for \textsc{MaxMatroidIntersection}, as well. A variant of the algorithm using circuit oracles is verified using the same methodology as before. 

\subparagraph*{Selecting an $S_X$-$T_X$ Path} There is a verified implementation of Breadth-First Search (BFS) that works on adjacency maps by Abdulaziz~\cite{graphAlgosBook}. It takes a set of sources and builds another adjacency map modelling a DAG of those edges that are actually explored by the search. If a vertex $v$ is reachable from a source $s$, there is a path $p$ in this DAG leading from a source $s'$ to $v$ such that $p$ is a shortest path in the actual graph leading from $s'$ to $v$. \texttt{find\_path} runs BFS for the sources $S_X$ and uses a DFS to obtain a path in the DAG and, hence, a shortest $S_X$-$T_X$ path.
For a concrete intersection problem and together with appropriate oracles, \texttt{find\_path} can be used to instantiate the locale of the intersection algorithm.

\subparagraph*{Bipartite Matching} An example for matroid intersection is \textit{maximum cardinality bipartite matching}. 
Consider a bipartite graph with edges $E$ between two disjoint sets of vertices $L$ and $R$.
For $e \in E$, we call the endpoint in $L$ the \textit{left} and the one in $R$ the \textit{right endpoint}.
We say that $M \subseteq E$ is independent w.r.t.\ $L$ iff no edges in $M$ share a left endpoint. Independence w.r.t.\ $R$ is defined analogously.
Valid \textit{matchings} are sets of vertex-disjoint edges in $E$ and are thus exactly those $M$ that are independent w.r.t.\ both $L$ and $R$. 
The two independence predicates satisfy the matroid axioms making this a matroid intersection problem.
Circuits w.r.t.\ $L$ or $R$ would be edges $e,d\in E$ sharing a left or right endpoint, respectively.

For $M \subseteq E$, we maintain a set data structure $[M]$, two maps $M_L$ and $M_R$, where $M_L$ associates every $x \in L$ with $e \in M$ for which $x$ is the left endpoint, and analogously for $M_R$. 
Weak independence and circuit oracles are very simple: If $M$ is independent w.r.t.\ $L$ ($R$), and $e \not \in M$, we check in $M_L$ ($M_R$) that $e$'s left (right) endpoint is not associated to any another edge.
If $M \cup \lbrace e \rbrace$ would be dependent w.r.t.\ $L$ or $R$, we would return the edges associated to $e$'s left or right endpoint as $\mathcal{C}_{L}(M,e)\setminus \lbrace e \rbrace$ or $\mathcal{C}_{R}(M,e)\setminus \lbrace e \rbrace$, respectively. We used these oracles to instantiate \textsc{MaxMatroidIntersection} to compute maximum matchings in bipartite graphs. The running time is $\mathcal{O} (\min \lbrace |L|, |R| \rbrace \cdot m \cdot (\log |R| + \log |L| + \log m))$ where the logarithms are due to using tree data structures, which is dominated by $\mathcal{O} (n \cdot m \cdot (\log n + \log m))$.
Without circuit oracles, we would have another multiplicative factor of $\min \lbrace |L|, |R| \rbrace$.

\section{Discussion}

We presented a formal analysis of algorithms to solve two optimisation problems for matroids: (a) computing the largest (in terms of accumulated weight) independent set in a given matroid and (b) computing the largest (in terms of cardinality) set that is independent w.r.t.\ two matroids.
The two algorithms are of great importance to practitioners and are implemented in a number of computer algebra systems, e.g.\ in Sage~\cite{sageMatroid} and Macaulay~\cite{macaulayMatroid}.
Additionally, we also briefly (due to space constraints) presented a formalisation of the analysis of an algorithm to solve an optimisation problem for greedoids, which are a generalisation of matroids.
In addition to formalising significant parts of the theory of matroids, greedoids, and the analysis of the algorithms, we showed that our approach can also be used to obtain practical executable verified implementations of important graph algorithms.
This demonstrates the potential practical role of matroids and greedoids to design well-factored formal mathematical libraries whereby proofs of multiple algorithms are done in one go that captures their mathematical essence, and that is later instantiated for different concrete examples.
The formalisation overall totalled 17.4K lines of proof script, with 11K lines dedicated to the theory of matroids and greedoids, 2.9K lines connecting graph theoretic problems to matroids and/or greedoids, and 3.5K lines on defining and verifying Kruskal's, Prim's, and the bi-partite matching algorithm.
Our formalisation is available at \url{https://doi.org/10.5281/zenodo.15758230}; it builds on and is part of an ongoing effort to formalise combinatorial optimisation in Isabelle/HOL~\cite{graphLibRepo}.

\subparagraph*{Related Work}
Most relevant to us is the formalisation of matroids in Isabelle/HOL by Keinholz~\cite{Matroids-AFP}.
Haslbeck et al.~\cite{Kruskal-AFP} formalised a classic proof of the correctness of the Best-In-Greedy Algorithm and obtained a verified, imperative implementation of Kruskal's algorithm.
They used a weak oracle based on a verified imperative union-find implementation by Löwenberg~\cite{loewenberg} to store the tree's connected components leading to their implementation having the best possible running time of $\mathcal{O}(m\alpha(n))$, where $\alpha$ is the inverse Ackermann function. 
Löwenberg also gives a functional version that we could use for our work, which would, however, not pay off in terms of running time.

In Lean, there is an ongoing effort by Nelson et al. to formalise matroid theory~\cite{lean_matroids}, including more advanced results about e.g.\ minors, isomorphism, Tutte's excluded minor theorem for finitary binary matroids, or Edmonds' Intersection Theorem (Lemma~\ref{lemma:ranks}), but (as of the time of writing) there seems to be no verified practical algorithms in that library.

Matroids were also formalised in Coq~\cite{MAGAUD2012406,Braun2024} and Mizar~\cite{Bancerek2008IntroductionTM}, however, both do not go very far.
Set systems were formalised in Isabelle/HOL by Edmonds and Paulson~\cite{EdmondsPaulsonCombDesign}, who did that in the context of verifying combinatorial designs.

Additionally, there are formalisations of matching algorithms~\cite{blossomIsabelleFull,RankingIsabelle} and Prim's algorithm~\cite{primsLammichNipkow}, in Isabelle/HOL, and of Dijkstra, Kruskal, and Prim~\cite{10.1007/978-3-030-81688-9_37}, in Coq, that are not based on matroids or greedoids, where all reasoning is done at the graph-theoretic level.
The Isabelle/HOL formalisation of Prim's algorithm~\cite{10.1007/978-3-030-81688-9_37} uses a priority queue to find edges to extend the spanning tree. We cannot use this more efficient (and standard) way due to conditions required to prove Theorem~\ref{greedoid:characterisation}.
Readers interested in comparing the style of formal reasoning when using matroids or greedoids, which is algebraic, to the one performed directly on graphs, which is more intuitive but combinatorial, should refer to those formalisations.

\subparagraph*{Future Work} There is one 'main' missing algorithm from our library which would be the most imminent piece of future work: the weighted version of the matroid intersection problem, whereby, instead of finding the largest cardinality set that is independent w.r.t.\ two matroids, one is to find the largest in terms of the weight of such a set, for a modular cost function.

Routine functions that check whether a matroid indeed satisfies the assumptions of a matroid, compute the dual of a matroid, etc., which are implemented in most computer algebra packages, are missing from our library and are an important next step for our work.
Naive implementations of those functions have a worst-case exponential running time, so utilising sophisticated enumeration algorithms~\cite{artOfProgrammingCombAlgos} would help with making that more practical.

\bibliography{long_paper,second_bib}
\end{document}